\documentclass[conference]{IEEEtran}
\usepackage{amssymb}
\usepackage{graphicx}
\usepackage{graphics}
\usepackage{mathrsfs}
\usepackage{amsmath}
\usepackage{amsthm}
\usepackage{amsfonts}
\usepackage{cite}

\def \hh{\vskip 0.5\baselineskip \hbox to \hsize}

\newtheorem{definition}{Definition}
\newtheorem{theorem}{Theorem}

\newcommand{\mb}{\mathbf}

\begin{document}

\title{Cascading Link Failure in the Power Grid: \\A Percolation-Based Analysis}
\author{ {\large Hongda Xiao and Edmund M. Yeh} \\
%EndAName
{\normalsize Department of Electrical Engineering, Yale University, New
Haven, CT 06511}\\
{\normalsize $\{$hongda.xiao, edmund.yeh$\}$@yale.edu} }
\maketitle

\begin{abstract}
Large-scale power blackouts caused by cascading failure are inflicting enormous socioeconomic costs. 
We study the problem of cascading link failures in power networks modelled by random geometric graphs
from a percolation-based viewpoint.  To reflect the fact that links fail according to the amount of power flow going through them, we introduce a model where links fail according to a probability which depends on the number of neighboring links.  We devise a mapping which maps links in a random geometric graph to nodes in a corresponding dual covering graph.  This mapping enables us to obtain the first-known analytical conditions on the existence and non-existence of a large component of operational
links after degree-dependent link failures.   Next, we present a simple but descriptive model
for cascading link failure, and use the degree-dependent link failure results to obtain the first-known analytical conditions on the 
existence and non-existence of cascading link failures. 
\end{abstract}

\section{Introduction}

In August 2003, a massive power blackout rolled across northeastern U.S. and Canada,
affecting some 55 million people.  The blackout resulted from a three-hour-long 
cascading failure triggered by an undetected initial power line failure in Ohio.  Later 
that year, 57 million Italians were left in the dark from a similar blackout.  Around
the world, the number of outages affecting large populations has steadily risen
in the past decade, causing enormous economic losses~\cite{Blackout}.  As
a consequence, the resilience and reliability of power systems have received
increasing attention.  Designing and constructing a smart, resilient,
and self-healing power grid has become a crucial challenge for 
governments and industry.  

The essential dynamic of power blackouts is a process whereby
the failure of a generating plant or a power line results in redistribution of the
power flow load onto other nearby generators or power lines.  If these other
plants and lines then fail due to excessive load, then this process can further
propagate and result in a {\em cascading failure}.  
A key question for the designer of the power grid is: 
can we predict whether a small number of initial failures will or will not trigger 
a cascading failure affecting
the entire network? 

Previous work has examined the problem of cascading failure from a number
of perspectives.  In~\cite{Dobson_initial}, the authors use a linear model for
power flows.  They model transmission line failures as the result of overflow,
and consider the dynamics of power networks in time. In \cite%
{Dobson_complex}, the authors include the growth of power demand, 
engineering responses to system failures, and the upgrading of generator
capacity. They characterize two
types of blackouts, one which is due to transmission lines reaching load
limits but incurring no line outages, the other which is due to multiple line
outages. They use simulations to illustrate how their proposed model fit the data
from American blackouts. 
%In \cite{Topology}, an experimental research of
%topology of power networks was proposed. The authors studied several
%connectivity properties of power networks, and compared connectivity of
%power networks with several other networks.
In \cite{Dobson_cascading}, the authors
characterize two kinds of critical operation points of the power system, 
due to the transmission line limit and the generator capacity limit,
respectively. They find by simulation that operation near critical
points could produce power law tails in the blackout size probability
distribution. In \cite{IntegerProgramming}, the authors propose a method to prevent 
cascading power failure by using integer programming techniques. They
consider the cost of upgrading the capacity of transmission lines and
calculate the optimal choice of link upgrade to prevent the expansion of
failure.   As seen by the brief review above, the current research on cascading
failures in power networks is primarily based on simulation and optimization algorithms. 
Due to the complex nature of the dynamics involved in cascading failures,
analytical results have been rather elusive. 

In~\cite{Zhening}, the authors propose a new approach to analyzing cascading failure
based on the theory of percolation~\cite{Gilbert,Grimmett,Penrose,interference,Zhening2}.  The perspective adopted in~\cite{Zhening} is as follows:
in a power network subject to redistribution of load, 
the presence of cascading failure can be assessed by whether or not a small number of initial node
(generator) or link (line) failures lead to 
a large connected component of failed nodes or links affecting the network {\em globally}. 
Percolation theory, which is concerned with assessing the global connectivity of networks,
provides the appropriate viewpoint with which to make this assessment. 
The authors of~\cite{Zhening} propose a simple but descriptive model for cascading node failure,
based on node susceptibility thresholds and degree-dependent node interactions.  Within the context of 
random geometric graphs,~\cite{Zhening} presents the first analytical conditions for the
existence and non-existence of cascading node failures.  

In this paper, we adopt the viewpoint presented in~\cite{Zhening}.  Instead of examining cascading failure
of nodes, however, we focus on cascading {\em link} failures.  This is motivated by the fact that blackouts in 
power grids are often triggered by link failures (line faults), which cause power flow redistribution onto neighboring
links, which when overburdened, can fail and propagate the cascade.   As in~\cite{Zhening}, we shall
focus on networks modelled by {\em random geometric graphs}, where
nodes are distributed spatially according to a Poisson point process with constant density $\lambda$, and
any two nodes within given distance $r$ of each other share a link.  In doing so,
we are motivated by the following consideration.  First, the topology of actual power grids is not readily
available.  Second, we use a random graph to model the ensemble topology of power grids. 
Third, there is evidence that geometry plays an important role in the topology of electrical power grids~\cite{Topology},
especially with respect to cascading dynamics.
Finally, our work is intended to provide insight into the dynamics of cascades within large-scale power networks, and
random geometric graphs provide a concrete setting in which analytical results may be obtained.  

This paper proceeds as follows.  After presenting some preliminary results on random
geometric graphs and continuum percolation, we study random geometric graphs in which each given link fails
(independently) with a probability which depends on the number of neighboring links (links which share an end
node with the given link).  This models the phenomenon in power networks where the probability of a link failure 
increases with its load 
in terms of power flow.   The larger the number of other links which share an 
end node with a given link, the more power flow is likely to traverse the given link,
and therefore the higher the probability that the given link fails. 
We present the first-known analytical conditions on the existence and non-existence of a large component of operational
links after degree-dependent link failures.  In order to establish these results, we devise a new map which maps links in a
random geometric graph to corresponding nodes in a dual covering graph.  Finally, we present a simple but descriptive model
for cascading link failure, and use the degree-dependent link failure results to obtain the first-known analytical conditions on the 
existence and non-existence of cascading link failures.

%The paper is organized as following. In Section \ref{graph}, we describe the
%graph model for the power grid. In Section \ref{nodepercolation}, we
%introduce the basic content of percolation theory and corresponding node
%failure process. In Section \ref{bondpercolation}, we introduce the basic
%content of link percolation theory and corresponding link failure process.
%In Section \ref{Dual}, we propose our technique to relate bond percolation
%process to site percolation process by using covering graph. In Section \ref%
%{Condition}, we propose our analytical results about resilience of power
%grid suffering i.i.d and degree-dependent link failure. In Section \ref%
%{cascading}, we propose the descriptive model for cascading link failure in
%power grid and relative analytical condition for resilience. In Section \ref%
%{conclusion}, we conclude the paper.

\section{Network Model}

\label{graph}

We use a random geometric graph to model the power network. A rigorous model is as follows.
%We consider networks modelled by random geometric graphs. That is, we assume that the network
%nodes are randomly distributed in space, and a link exists between two
%nodes if the distance between them is sufficiently small.  
Let $\mathcal{H}_{\lambda}$ be 
a homogeneous Poisson point process with density $\lambda > 0$ in $\mathbb{R}^2$.    Let
 $G(\mathcal{H}_{\lambda},r)$ be the (infinite) undirected graph with vertex set $\mathcal{H}_{\lambda}$
and with undirected links connecting all pairs $\{{\mb X}_i, {\mb X}_j\}$ with $\|{\mb X}_i- {\mb
X}_j \|\leq r$ $(r>0)$ where ${\mb X}_i, {\mb X}_j \in \mathcal{H}_{\lambda}$. 
Due to the scaling property of random geometric
graphs~\cite{Penrose}, in the following, we focus on $G(\mathcal{H}_{\lambda},1)$.

A fundamental result for random geometric graphs 
concerns a phase transition effect whereby the macroscopic behavior of the network is very different for
densities below and above some critical value $\lambda_c$. For $\lambda<\lambda_c$ (subcritical or non-percolated), the
connected component containing the origin contains a finite number of points almost surely.
For $\lambda>\lambda_c$ (supercritical or percolated), the connected component containing the origin 
contains an infinite number of points with a positive probability~\cite{Gilbert, Penrose, Grimmett}.

Let $\mathcal{H}_{\lambda,\mathbf{0}}=\mathcal{H}_{\lambda}\cup \{\mathbf{0}\}$, i.e., the union
of the origin and the homogeneous Poisson point process with density $\lambda$.\footnote{Note that
in a random geometric graph induced by a homogeneous Poisson point process, the choice of the
origin can be arbitrary.} As discussed, a phase transition takes place at the critical
density. 

\vspace{0.1in}%
\begin{definition} For $G(\mathcal{H}_{\lambda,\mathbf{0}},1)$, the
\emph{percolation probability} $p_{\infty}(\lambda)$ is the probability that the connected component
containing the origin has an infinite number of nodes of the graph. The \emph{critical density}
$\lambda_c$ is defined as
\begin{equation}
\lambda_c=\inf \{\lambda>0: p_{\infty}(\lambda)>0\}.
\end{equation}
\end{definition}
\vspace{0.1in}%

It is known that if $\lambda>\lambda_c$, then there exists a unique infinite component in $G(\mathcal{H}_{\lambda},1)$.
A fundamental result of continuum percolation states that $0<\lambda_c<\infty$~\cite{MeRo96}. Exact values of
$\lambda_c$ and $p_{\infty}(\lambda)$ are not yet known. Simulation studies show that~$1.43<\lambda_c<1.44$
%\cite{QuToZi00}.  
%In the rest of the paper, we shall use the terms {\em percolated (non-percolated)} and {\em supercritical (subcritical)}
%interchangeably, to mean the existence (non-existence) of an infinite component in $G(\mathcal{H}_{\lambda},1)$.

We have defined the random geometric graph in the infinite graph setting.  The finite analog of $G(\mathcal{H}_{\lambda},1)$
can be defined as follows~\cite{Penrose}.  Let ${\cal X}_n = \{{\mathbf X}_1, \ldots, {\mathbf X}_n\}$ be $n$ points independently and uniformly
distributed in a square of area $n/\lambda$ in ${\mathbb R}^2$.  Form a graph $G({\cal X}_n, 1)$ by placing a link between any 
two nodes $i$ and $j$ such that $\|{\mathbf X}_i - {\mathbf X}_j\| \leq 1$.  Let $L_1(G({\cal X}_n, 1))$ be the size of the
largest connected component in $G({\cal X}_n, 1)$.  Then it can be shown that as $n \rightarrow \infty$ with $\lambda$ fixed, 
$n^{-1} L_1(G({\cal X}_n, 1)) \rightarrow p_{\infty}(\lambda)$ in probability.  Thus, if $\lambda>\lambda_c$, there exists
a unique largest component of size $\Theta(n)$ in $G({\cal X}_n, 1)$.  On the other hand, it can be shown that if
$\lambda< \lambda_c$,  then the largest component in $G({\cal X}_n, 1)$ can only have size $O(\log n)$~\cite{Penrose}.

While finite random geometric graphs are clearly of more interest for practical applications, it is more convenient to 
consider infinite graphs for expositional purposes.  For this reason, we will concentrate on infinite graphs in the rest of the paper,
while keeping in mind the correspondence between the finite and infinite graphs.

\section{Degree-dependent Random Link
Failures}

In the power grid, large-scale blackouts are often triggered by power line faults.  This was 
the case in the massive 2003 North American blackout.  In our context, line faults can be modeled
as link failures.  Thus, in order to further understand blackouts in the power grid, it is important
to characterize the resilience of large-scale networks to link failures. 
  
In power networks, the probability of link failure typically increases with its load 
in terms of power flow.   The larger the number of other links which share an 
end node with a given link, the more power flow is likely to traverse the given link,
and therefore the higher the probability that the given link fails.  To model this effect,
we consider degree-dependent random link failures.  We first define the degree of a link.

\begin{definition}
The degree of link $(i,j)$, denoted by $d(i,j)$, is the number of other links which share
an end vertex with $(i,j)$.  That is, 
$$	d(i,j) = d_i + d_j -2, $$
where $d_i$ and $d_j$ are the node degrees of $i$ and $j$. 
\label{def:degree}
\end{definition}

Consider a scenario where each {\em link} in $G(\mathcal{H}_{\lambda },1)$ fails (independently) with a probability $q(k)$ that depends
on the link's degree $k$.  Let $G(\mathcal{H}_{\lambda },1,q(\cdot))$ be the remaining graph after degree-dependent
link failures.  In order to study the connectivity properties of $G(\mathcal{H}_{\lambda },1,q(\cdot))$, we construct a mapping
where the bond percolation process in $G(\mathcal{H}_{\lambda },1)$ is mapped to a site percolation process in
a {\em covering graph} $G_c(\mathcal{H}_{\lambda },1)$.

%\subsubsection{Dual Covering Graph}

%\label{Dual}

%We consider a dual graph $G_{c}(\mathcal{H}_{\lambda },1)$ of original graph
%$G(\mathcal{H}_{\lambda },1)$, which can transfer bond percolation process
%of $G(\mathcal{H}_{\lambda },1)$ to site percolation process of $G_{c}(%
%\mathcal{H}_{\lambda },1)$. We call $G_{c}(\mathcal{H}_{\lambda },1)$ the
%\textit{covering graph} of original graph $G(\mathcal{H}_{\lambda },1)$.

The covering graph of a graph $G$ is defined as follows.

\begin{definition}
$G_{c}(V,E)$ is called the \textit{covering graph} of $G(V,E)$ if
%\newcounter{cover}
%\begin{list}{\bfseries\upshape \arabic{cover}.}
%{\usecounter{cover}}
\begin{enumerate}
\item Each node in $G_c$ denotes a link in $G$, and each link in $G$ is denoted by a unique node in $G_c$.
\item Two nodes in $G_c$ share a link if and only if the corresponding links in $G$ share a common end vertex (in $G$).
\end{enumerate}
\label{def:covering}
\end{definition}

Comparing Definitions~\ref{def:degree} and~\ref{def:covering}, we have
the convenient fact that the degree of a link in $G$ is equal to the degree of its corresponding node
in $G_c$.  

%In \cite{Grimmett}, it is proved that the bond percolation process of $G$ is
%equivalent to the site percolation process of its covering graph $G_c$. 

It should be clear that the covering graph of a given graph is not unique.
Moreover, in the context of geometric graphs, 
most covering graphs are of little importance, for they hold only
the connection information, losing all the geometric information.  In order
to utilize the geometric information in the original graph, we need to place
additional constraints on the geometric location of nodes in the covering graph.
Fortunately, we find a construction for the covering graph which is particularly
useful for random geometric graphs with random link failure.  In this 
construction, the location of nodes in the covering graph is directly determined
by the node locations in the original graph, thereby maintaining the geometric
information. 

Given $G(\mathcal{H}_\lambda, 1)$, we construct the covering graph
$G_c( \mathcal{H}_\lambda, 1)$ as follows.
\begin{enumerate}
\item Consider the circles with radius $1/2$ and the nodes in $G( \mathcal{H}_\lambda, 1)$ as the centers.
\item Let the mid-points of each link in $G( \mathcal{H}_\lambda, 1)$ be the nodes in $G_c( \mathcal{H}_\lambda, 1)$. 
\item Two nodes in $G_c( \mathcal{H}_\lambda, 1)$ share a link if and only if they lie in the same circle of radius 1/2 in $G( \mathcal{H}_\lambda, 1)$.
\end{enumerate}
The construction is illustrated in Figure~\ref{fig:covering}.
\begin{figure}[h]
\centering
\includegraphics[scale = 0.7]{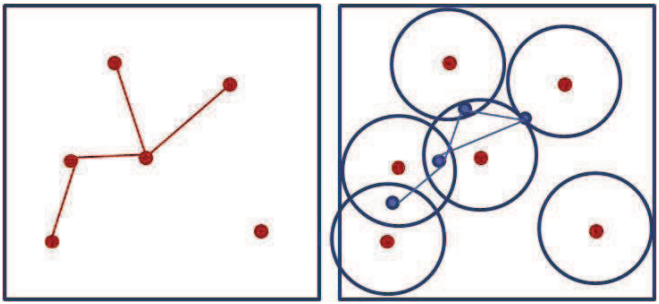}
\caption{Construction of the covering graph: the red nodes and links belong to $G( \mathcal{H}_\lambda, 1)$;
the blue nodes and links belong to $G_c( \mathcal{H}_\lambda, 1)$.}
\label{fig:covering} 
\end{figure}

\begin{theorem}
The graph $G_c( \mathcal{H}_\lambda, 1)$ constructed above is a
covering graph of $G( \mathcal{H}_\lambda, 1)$.
\end{theorem}

\begin{proof}
First, each link in $G(\mathcal{H}_{\lambda },1)$ corresponds to a unique
node in $G_{c}(\mathcal{H}_{\lambda },1)$, which is the link's
midpoint. Second, since the length of links in $G(\mathcal{H}_{\lambda
},1) $ is less than or equal to $1$, two nodes in $G_{c}(\mathcal{H}_{\lambda
},1)$ lie in the same circle (of radius 1/2) in $G( \mathcal{H}_\lambda, 1)$
if and only if their corresponding links in $G(\mathcal{H}_{\lambda },1)$ 
share an end vertex. Thus $G_{c}(\mathcal{H}
_{\lambda },1)$ is a covering graph of $G(\mathcal{H}_{\lambda },1)$.
\end{proof}

One important consequence of the above construction is that 
the nodes of the covering graph $G_{c}(\mathcal{H}_{\lambda },1)$ is also
driven by a Poisson point process, albeit of a different density. 

\begin{theorem}
\label{poisson} 
The nodes of $G_{c}(\mathcal{H}_{\lambda },1)$ are distributed
according to a Poisson point process with density $\lambda
^{\prime }=\frac{\pi \lambda ^{2}}{2}$.
\end{theorem}

\begin{proof}
To prove that the nodes in $G_{c}(\mathcal{H}_{\lambda },1)$ are driven
by a Poisson process, it is most convenient to focus on the finite
random geometric graph $G({\cal X}_n, 1)$ where ${\cal X}_n \equiv \{{\mathbf X}_1, \ldots, {\mathbf X}_n\}$ 
are $n$ points independently and uniformly
distributed in a square ${\cal A}$ of area $n/\lambda$ in ${\mathbb R}^2$.  
Now consider the covering graph $G_c({\cal X}_n, 1)$ of $G({\cal X}_n, 1)$ 
 obtained by
the construction above.  For any $i=1, \ldots, n$, 
conditioned on the location of ${\mb X}_{i}$, 
the midpoint of any link between ${\mb X}_{i}$ and ${\mb X}_{j}$ in $G({\cal X}_n, 1)$
is uniformly distributed in the circle centered at ${\mb X}_{i}$ with radius $1/2$.
Since the ${\mb X}_i$'s are independently and uniformly distributed in
${\cal A}$, so are the midpoints, which are the same as the nodes of the 
covering graph $G_c({\cal X}_n, 1)$.

To calculate the density of $G_c({\cal X}_n, 1)$, we note that
the mean degree of any node ${\mb X}_i$ in $G({\cal X}_n, 1)$
(ignoring border effects which diminish as $n$ tends to 
infinity) is $\pi \lambda $.   The summation of the mean degrees
of all nodes  ${\mb X}_{1}, {\mb X}_{2},..., {\mb X}_{n}$ equals $n\pi
\lambda$.   Now the number of edges of the
graph should equal $1/2$ of the sumation of the degrees of all the nodes.
Thus the mean number of edges in $G({\cal X}_n, 1)$ equals $n\pi
\lambda/2$. Therefore, the density of the links in
$G({\cal X}_n, 1)$, which is
also the density of the nodes in the covering graph $G_c({\cal X}_n, 1)$,
is given by 
$$ \lambda ^{\prime } = \frac{n \pi \lambda /2}{n/\lambda} = \frac{\pi \lambda^2}{2}. $$
The theorem now follows by taking the limit as $n \rightarrow \infty$
with $\lambda$ fixed.
\end{proof}

It is important to remark that even though the nodes of the covering graph
$G_{c}(\mathcal{H}_{\lambda },1)$ are distributed according to a Poisson
process, $G_{c}(\mathcal{H}_{\lambda },1)$ is {\em not} a random geometric 
graph.  To see this, note that two nodes in $G_{c}(\mathcal{H}_{\lambda },1)$ 
being within distance 1 of each other does not imply that there exists a 
link between the nodes.  Nevertheless, $G_{c}(\mathcal{H}_{\lambda },1)$ is
a subgraph of a random geometric graph with the same node locations as 
$G_{c}(\mathcal{H}_{\lambda },1)$. 

%\section{{\protect\large {Sufficient and Necessary Condition of Existence of
%Infinite Connected Component}}}

%\label{Condition}

We now derive analytical conditions on the existence and non-existence of
an infinite connected component of operational nodes in $G(\mathcal{H}_{\lambda },1,q(\cdot))$
after degree-dependent link failures in of $G(\mathcal{H}_{\lambda },1)$.
The key techniques we use are the mapping from a bond percolation 
process in the original graph to a site percolation in the covering graph,
and that link failures affect network connectivity less than node failures
at the same level.  

Our first main result gives a sufficient condition for there to exist
an infinite connected component of operational nodes in $G(\mathcal{H}_{\lambda },1,q(\cdot))$
after each link in  $G(\mathcal{H}_{\lambda },1)$ fails (independently) according to
a degree-dependent probability $q(k)$, where $k$ is the degree of the link, as given
in Definition~\ref{def:degree}.

%\begin{lemma}\label{Site and Bond}
%for a random geometric graph $G(\mathcal{H}_{\lambda },1)$ with $\lambda > \lambda_r > \lambda_c$, if it is percolated after a node failure process in which each node's failure probability is
%\end{lemma}

%\subsection{\textsl{{\protect\large {i.i.d link failures}}}}

%\subsection{\textsl{{\protect\large {degree-dependent link failures}}}}

%Now we consider degree-dependent link failures. We think that failure
%probability of a node is $q(k)$, where k is the summation of the degree the
%link. We use $G(\mathcal{H}_{\lambda },1,q(k))$ and $G_{c}(\mathcal{H}%
%_{\lambda },1,q(k))$ to indicate the remaining graph and its covering graph,
%respectively.

\begin{theorem}
\label{sufficient} 
For any $\lambda_1>\lambda_c$ and $G(\mathcal{H}_\lambda,1)$ with $\lambda>\lambda_1$, there
exists $k_1< \infty$ which depends on $\lambda$ and $\lambda_1$, such that if
\begin{equation}\label{qk-upper-bound}
q(k)\leq 1-\frac{\lambda_1}{\lambda}, \quad \mbox{for all}~1\leq k \leq k_1,
\end{equation}
then with probability 1, there exists an infinite connected component in
$G(\mathcal{H}_\lambda,1,q(\cdot))$.
\label{thm:sufficient}
\end{theorem}

An interesting implication of Theorem \ref{thm:sufficient} is that even if
all links with degree larger than $k_1$ fail with probability 1, an infinite component still
exists in the remaining graph as long as~\eqref{qk-upper-bound} is satisfied.

\begin{proof}
The proof uses the concept of the covering graph, the relationship
between bond percolation and site percolation, as well as
key ideas from the proof of Theorem 1(i) in~\cite{Zhening}.
Due to space constraints, we shall only outline the key steps.

Consider two random failure models in $G(\mathcal{H}_{\lambda },1)$.  In the
first model, each node fails (with
all associated links) independently with probability $1-\frac{\lambda_1}{\lambda}$. Let
$G^{site}_1(\mathcal{H}_\lambda, 1)$ be the remaining graph.   In the second 
model, each {\em link} fails independently with probability $1-\frac{\lambda_1}{\lambda}$.
Let $G_1(\mathcal{H}_\lambda, 1)$ be the remaining graph.  Since 
 $\lambda_1>\lambda_c$, $G^{site}_1(\mathcal{H}_\lambda, 1)$ is a random geometric
 graph in the supercritical regime.  Let 
$G_{c}(\mathcal{H}_{\lambda },1,q(\cdot))$ and $
G_{1c}(\mathcal{H}_{\lambda },1)$ be the covering graphs (according
to the construction given above) of $G(\mathcal{H}_{\lambda },1,q(\cdot))$ 
and $G_{1}(\mathcal{H}_{\lambda },1)$, respectively.

In the following, we use the square lattice construction from~\cite{Zhening}.  Please
refer to~\cite{Zhening} for further details on the construction. 
Consider the square lattice $\mathcal{L}=d\cdot\mathbb{Z}^2$, where
$d$ is the edge length. 
As in~\cite{Zhening}, define event $A^{site}_a(d)$ for each horizontal edge $a$ in $\mathcal{L}$ as the set of outcomes for which
the rectangle $R_a$ is crossed\footnote{Here, a rectangle
$R=[x_1,x_2]\times[y_1,y_2]$ being crossed from left to right by a connected component in
$G^{site}_1(\mathcal{H}_\lambda,1)$ means that there exists a sequence of nodes $v_1,v_2,...,v_m\in
G^{site}_1(\mathcal{H}_\lambda,1)$ contained in $R$, with $||\mathbf{x}_{v_i}-\mathbf{x}_{v_{i+1}}||\leq
1, i=1,...,m-1$, and $0<x(v_1)-x_1<1/2, 0<x_2-x(v_m)<1/2$, where $x(v_1)$ and $x(v_m)$ are the
$x$-coordinates of nodes $v_1$ and $v_m$, respectively.  A rectangle being crossed from top to
bottom is defined analogously.} from left to right by a connected component in
$G^{site}_1(\mathcal{H}_\lambda,1)$.  If $A^{site}_a(d)$ occurs, we say that rectangle $R_a$ is a {\em good}
rectangle, and edge $a$ is a {\em good} edge. Let
$ p^{site}_g(d)\triangleq\Pr(A^{site}_a(d)).$  
Define $A^{site}_a(d)$ similarly for all vertical edges by rotating the rectangle by $90^{\circ}$. 

%\begin{figure}[h]
%\label{A_a^{site}} \centering
%\includegraphics[scale = 0.33]{A_a}
%\caption{An example of the event $A_a^{site}(d)$}
%\end{figure}

Further define event $B^{site}_a(d)$ for edge $a$ in $\mathcal{L}$ as the set of outcomes for which 
(1) $A^{site}_a(d)$ occurs, and (2) the left square $S_a^-$ and the right square $S_a^+$ are
both are both crossed from top to bottom by connected components
in $G^{site}_1(\mathcal{H}_\lambda,1)$.
If $B^{site}_a(d)$ occurs, we say that rectangle $R_a$ is a {\em complete} rectangle, and edge $a$ is a
{\em complete} edge. Let
$
p^{site}_c(d)\triangleq\Pr(B^{site}_a(d)).
$
Define $B^{site}_a(d)$ similarly for all vertical edges by rotating the rectangle by $90^{\circ}$. 
An example of a complete rectangle is illustrated in
Figure~\ref{fig:CompleteRectangle}.

\begin{figure}[h]
\label{B_a} \centering
\includegraphics[scale = 0.33]{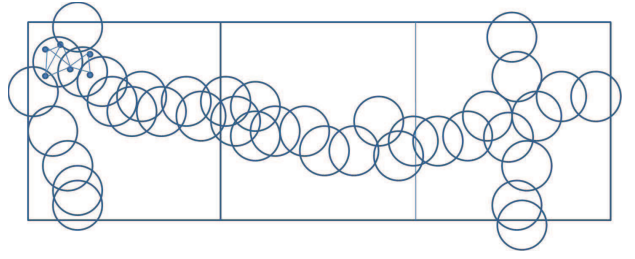}
\caption{An example of a complete rectangle.  The circles have radius 1/2.}
\label{fig:CompleteRectangle}
\end{figure}

%As in~\cite{Zhening}, we can use the FKG inequality to show that 
%$ p_{c}^{site}(d) \geq (p_{g}^{site}(d))^{3}$.  

We further define complete events $\{B'_{a}(d)\}$ 
with respect to $G_{1}(\mathcal{H}_{\lambda },1)$
in the same way as we defined 
complete events $\{B^{site}_{a}(d)\}$
with respect to $G^{site}_1(\mathcal{H}_\lambda, 1)$.   Similarly, define 
complete events $\{B''_{a}(d)\}$ with respect to
$G_{1c}(\mathcal{H}_{\lambda },1)$, the covering graph of 
$G_{1}(\mathcal{H}_{\lambda },1)$, and define 
complete events $\{B_{a}(d)\}$ with respect to
$G_{c}(\mathcal{H}_{\lambda },1,q(\cdot))$.

We shall now carry out a series of stochastic couplings to obtain our result. 
Since the critical link failure probability (for i.i.d. link failures) is greater than or equal to  
critical node failure probability (for i.i.d. node failures) in $G(\mathcal{H}_{\lambda },1)$,
we can stochastically couple $G_{1}(\mathcal{H}_{\lambda },1)$ and  
$G^{site}_{1}(\mathcal{H}_{\lambda },1)$ such that the existence of crossings 
defined in events $\{B^{site}_{a}(d)\}$ for $G^{site}_{1}(\mathcal{H}_{\lambda },1)$
implies the existence of crossings defined in events 
$\{B'_{a}(d)\}$ for $G_{1}(\mathcal{H}_{\lambda },1)$.  Next, by using the construction of
the covering graph $G_{1c}(\mathcal{H}_{\lambda },1)$, it is straightforward to verify
that the existence of crossings defined in events 
$\{B'_{a}(d)\}$ for $G_{1}(\mathcal{H}_{\lambda },1)$ implies
the existence of crossings defined in events 
$\{B''_{a}(d)\}$ for $G_{1c}(\mathcal{H}_{\lambda },1)$.\footnote{To be
absolutely precise, the definition of a crossing needs to be altered slightly here,
requiring the $x$-coordinates of the first and last node in the rectangle 
to be within 1 of the boundary, rather than 1/2.  However, this alteration 
does not affect our results below.}  

Let $p''_{c}(d) \equiv Pr(B''_{a}(d))$.  By the coupling argument above, we have
$p''_{c}(d) \geq p^{site}_{c}(d)$.  From~\cite{Zhening}, we know that 
 $p^{site}_c(d)$ converges to 1 as
$d\rightarrow \infty$ when $G^{site}_1(\mathcal{H}_{\lambda},1)$ is in the supercritical phase. 
Thus,  $p''_c(d)$ also converges to 1 as $d\rightarrow \infty$.

Define
\begin{small}
\begin{equation}\label{d-epsilon}
d(\lambda',\lambda_1)\triangleq
\inf\left\{d>4:p''_c(d)-\frac{1}{\left(\frac{d}{2}+2\right)\left(\frac{3d}{2}+2\right)
\lambda'}>1-q_0\right\},
\end{equation}
\end{small}
where $\lambda'$ is the density of $G_{c}(\mathcal{H}_{\lambda
},1)$, the covering graph of $G(\mathcal{H}_{\lambda },1)$, and 
$q_0\triangleq \frac{1}{9+2\sqrt{3}}$. Now choose the edge length of $\mathcal{L}$ as
$d=d(\lambda',\lambda_1)$.

Define event $C_a(d)$ for each horizontal edge $a$ in $\mathcal{L}$ as the set of outcomes for
which the number of nodes of $G_c(\mathcal{H}_\lambda, 1)$ in
$R_a'$ is strictly less than
\begin{equation}\label{N-1}
k_1\triangleq
2\left(\frac{d(\lambda',\lambda_1)}{2}+2\right)\left(\frac{3d(\lambda',\lambda_1)}{2}+2\right)\lambda',
\end{equation}
where 
$R_a'$ is the rectangle $R_a$ extended by 1 in all directions.
Define $C_a(d)$ similarly for all vertical edges by rotating the rectangle by $90^{\circ}$. If
$C_a(d)$ occurs, we call rectangle $R_a$ and edge $a$ {\em efficient}. Let
$
p_e(d)\triangleq\Pr(C_a(d)).
$

We say an edge $a$ in $\mathcal{L}$ is {\em open} if and only if it is both complete with respect to
$G_{1c}(\mathcal{H}_{\lambda},1)$ and efficient with respect to $G_{c}(\mathcal{H}_{\lambda
},1)$, i.e. when events $B''_a(d)$ and $C_a(d)$ both occur.  An edge is {\em closed} otherwise.

%\begin{figure}[h]
%\label{R_a} \centering
%\includegraphics[scale=0.45]{R_a}
%\caption{The rectangle $R_a^\prime$}
%\end{figure}

When $C_a(d)$ occurs for edge $a$ in $\mathcal{L}$, no node of $G_c(\mathcal{H}_\lambda,1,q(\cdot))$
in $R_a$ has degree strictly greater than $k_1$ in $G_c(\mathcal{H}_\lambda,1)$. 
In addition, if $q(k)$
satisfies (\ref{qk-upper-bound}), a link in $G(\mathcal{H}_\lambda,1)$ with degree $k, 1\leq k\leq k_1$
(or equivalently a node in $G_{c}(\mathcal{H}_{\lambda },1)$ with
degree $k, 1\leq k\leq k_{1}$), survives with a probability greater than or equal to
$\frac{\lambda_1}{\lambda}$ in the degree-dependent failures model. 
On the other
hand, for the independent random failures model, a link in $G(\mathcal{H}_\lambda,1)$
(or equivalently a node in $G_{c}(\mathcal{H}_{\lambda },1)$) survives with probability
exactly equal to $\frac{\lambda _{1}}{\lambda }$.   
Thus we can
couple $G_c(\mathcal{H}_\lambda,1,q(\cdot))$ with $G_{1c}(\mathcal{H}_\lambda, 1)$ so that the
existence of crossings defined in events $\{B''_a(d)\}$ for $G_{1c}(\mathcal{H}_\lambda, 1)$ implies
the existence of crossings defined in events $\{B_a(d)\}$ for
$G_c(\mathcal{H}_\lambda,1,q(\cdot))$.
Therefore, a path of open edges in $\mathcal{L}$ implies a
connected component in $G_c(\mathcal{H}_\lambda,1,q(\cdot))$.

As in~\cite{Zhening}, we find 
\begin{equation}
p_o(d) \triangleq Pr(B''_a(d) \cap C_a)
\geq p''_c(d) + p_e(d) -1,
\label{p-open}
\end{equation}
and 
\begin{equation}
p_e(d) \geq 1-\frac{1}{\left(\frac{d(\lambda',\lambda_1)}{2}+2\right)
\left(\frac{3d(\lambda',\lambda_1)}{2}+2\right)\lambda'}.
\label{p-efficient-bound}
\end{equation}
By \eqref{d-epsilon}, \eqref{p-open} and \eqref{p-efficient-bound}, we have
\begin{equation}
p_o(d) \geq p_c(d)+p_e(d)-1 > 1-q_0.
\end{equation}

Finally, we use the dual lattice technique in \cite{Zhening} to prove the
existence of an infinite open edge cluster in $\mathcal{L}$, which 
implies the existence of an infinite connected component in 
$G_{c}(\mathcal{H}_{\lambda},1,q(\cdot))$, and therefore
an infinite connected component in 
$G(\mathcal{H}_{\lambda},1,q(\cdot))$.  
\end{proof}

Theorem~\ref{thm:sufficient} provides a sufficient condition
for $G(\mathcal{H}_\lambda,1,q(\cdot))$ to have an infinite component.  The next theorem
provides a sufficient condition for
$G(\mathcal{H}_\lambda,1,q(\cdot))$ to have {\em no} infinite component. Thus, it provides a {\em
necessary} condition for $G(\mathcal{H}_\lambda,1,q(\cdot))$ to have an infinite component. 
Here again, the concept of the covering graph allows us to leverage Theorem 1(ii) of~\cite{Zhening} to prove our
new result.

\begin{theorem}
Given $G(\mathcal{H}_\lambda,1)$ with $\lambda>\lambda_c$, let $\lambda' = \pi \lambda^2/2$
be the density of the covering graph $G_c(\mathcal{H}_\lambda,1)$.  
If either
\begin{equation}\label{qk-lower-bound-1} e^{-\frac{\lambda'}{2}}+\sum_{k=1}^{\infty}
\frac{(\frac{\lambda'}{2})^k}{k!}e^{-\frac{\lambda'}{2}}q(k-1)^k>1-\frac{1}{27}
\end{equation}
when $q(k)$ is non-decreasing in $k$, or if
\begin{multline}\label{qk-lower-bound-2}
\sum_{k=1}^{\infty}\frac{\left(\frac{\lambda'}{2}\right)^k}{k!}e^{-\frac{\lambda'}{2}}\Big(
\sum_{m=0}^{\infty}\frac{[\lambda'(2\sqrt{2}+\pi)]^m}{m!}e^{-\lambda'(2\sqrt{2}+\pi)}\times\\
\left(1-q(m+k-1)^k\right)\Big)<\frac{1}{27}
\end{multline}
when $q(k)$ is non-increasing in $k$, then with probability 1, there is no infinite connected
component in $G(\mathcal{H}_\lambda,1,q(\cdot))$.
\label{thm:necessary}
\end{theorem}

\begin{proof}
The existence of an infinite connected component in 
$G(\mathcal{H}_{\lambda },1,q(\cdot))$ after degree-dependent link failures in 
$G(\mathcal{H}_{\lambda },1)$
according to failure probability $q(k)$ (where $k$ is the link degree) 
is equivalent to the existence of an infinite
connected component in the covering graph $G_{c}(\mathcal{H}_{\lambda
},1,q(k))$ after degree-dependent node failures in $G_{c}(\mathcal{H}_{\lambda },1)$
according to failiure probability $q(k)$ (where $k$ is the node degree). 

In $G_{c}(\mathcal{H}_{\lambda },1)$, the nodes are Poisson distributed
with density $\lambda'$.  The length of any link is at most 1, and any two
nodes with distance strictly larger than 1 do not share a link 
(although two nodes with distance less than 1 may not necessarily 
share a link).  
Now consider the random geometric graph $G_{r}(
\mathcal{H}_{\lambda ^{\prime }},1)$ with density $\lambda ^{\prime }$ 
where the node locations are the same as those for 
$G_{c}(\mathcal{H}_{\lambda },1)$.   Let $G_{r}(
\mathcal{H}_{\lambda ^{\prime }},1,q(\cdot))$ denote the remaining graph 
after degree-dependent node failures in 
$G_{r}(\mathcal{H}_{\lambda ^{\prime }},1)$ according to failure
probability $q(k)$.  Since $G_{c}(\mathcal{H}_{\lambda },1)$ is 
a subgraph of $G_{r}(\mathcal{H}_{\lambda ^{\prime }},1)$,
we can stochastically couple $G_{r}(\mathcal{H}_{\lambda },1,q(\cdot))$
and $G_{c}(\mathcal{H}_{\lambda },1,q(\cdot))$
such that if there does not exists an infinite
connected component in $G_{r}(\mathcal{H}_{\lambda },1,q(\cdot))$, then
there also does not exist an infinite
connected component in $G_{c}(\mathcal{H}_{\lambda },1,q(\cdot))$,
By Theorem 1(ii) of \cite{Zhening}, if either~(\ref{qk-lower-bound-1}) or
(\ref{qk-lower-bound-2}) is satisfied, then 
$G_{r}(\mathcal{H}_{\lambda },1,q(\cdot))$ does not have an infinite
component, and therefore $G_{c}(\mathcal{H}_{\lambda },1,q(\cdot))$
does not have an infinite component either. 
\end{proof}

\section{{\protect\large {Cascading Link Failure}}}

\label{cascading}

Cascading failure (blackouts) in power networks often result from a small number
of initial line faults triggering many more line failures affecting the entire network. 
To understand the underlying dynamics of such phenomena, we focus on 
cascading link failures in large-scale networks modelled by random geometric graphs.

Consider a network modelled by a random geometric graph
$G(\mathcal{H}_{\lambda},1)$ with $\lambda>\lambda_c$, where an initial failure seed is represented by a single failed
link. We are
interested in whether this initial small shock can lead to a global cascade of link failures, defined as follows.

\begin{definition}
A cascading link failure is an ordered sequence of link failures triggered by an
initial failure seed resulting in an infinite component of failed links in
the networks.
\end{definition}

%There is an example of link cascading failure of a component.
%\begin{figure}[h]
%\label{cascade} \centering
%\includegraphics[scale = 0.7]{cascade}
%\caption{An example of cascading failure}
%\end{figure}

To describe cascading failures, we use the following simple but descriptive model similar to
the model proposed in~\cite{Zhening} for cascading node failures.   Assume that there
is a ``susceptibility threshold" $\psi_l \in [0,1]$ associated with each link $l$, 
where the $\psi_l$'s are i.i.d. random variables with probability density function $f(\psi)$. 
Due to local redistribution of power flow load after link failures, each link $l$ fails if a given fraction $\psi_l$ of its
neighboring links (the links which sharing an end node with $l$)
have failed.  The order of the failure sequence is then the topological order determined by
the location of the initial link failure and the threshold $\psi_l$ of each node $l$.   

Note that the initial link failure can grow only when some neighboring link, say $m$, of the initial
failure seed has a threshold satisfying $\psi_m\leq \frac{1}{k_m}$, where $k_m\geq 1$ is the
link degree of $m$.  Such a link is called \emph{vulnerable}. The probability of a link being vulnerable
is
\begin{equation}\label{rho-k}
\rho_k = F_{\psi}\left(\frac{1}{k}\right)=\int_{0}^{\frac{1}{k}} f(\psi)d\psi,
\end{equation}
where $F_{\psi}(\cdot)$ is the cumulative distribution function of $\psi_l$. 
When the initial failure seed is directly connected to a component of vulnerable links, all links in this component
fail.  Thus, a cascade of failed links forms
when the network has an infinite component of vulnerable links and the initial failure seed is either inside this
component or adjacent to (share a link with) some link in this component.

On the other hand, if link $l$ has a threshold satisfying $\psi_l> \frac{k_l-1}{k_l}$, where $k_l$
is the link degree of $l$, then link $l$ will not fail as long as at least one neighboring link is operational.
Such a link is called \emph{reliable}. Otherwise, if $\psi_l \leq \frac{k_l-1}{k_l}$,  we call link
$l$ \emph{unreliable}. For $k\geq 1$, the probability of a node being reliable is given by
\begin{equation}\label{sigma-k}
\sigma_k = 1-F_{\psi}\left(\frac{k-1}{k}\right)=\int_{\frac{k-1}{k}}^1 f(\psi)d\psi.
\end{equation}
Note that when two reliable links share an end node and neither is an initial
failure seed, no matter what else happens in the network, they remain operational.

The following two theorems present our main results on cascading link failure in random geometric graphs.
The results depend crucially on the degree-dependent link failure results in Theorems~\ref{thm:sufficient}
and~\ref{thm:necessary}.  

\begin{theorem}
For any $\lambda_1>\lambda_c$ and $G(\mathcal{H}_{\lambda},1)$
with $\lambda>\lambda_1$, there exists $k_1<\infty$ depending on $\lambda$ and $\lambda_1$ such
that if
\begin{equation}\label{eq:F-psi}
F_{\psi}\left(\frac{1}{k_1}\right) \geq\frac{\lambda_1}{\lambda},
\end{equation}
then with probability 1, there exists an infinite component of vulnerable links in
$G(\mathcal{H}_{\lambda},1)$. Moreover, if the initial link failure is inside this component or
adjacent to some link in this component, then with probability 1, there is a cascading link failure in
$G(\mathcal{H}_{\lambda},1)$.
\end{theorem}
\begin{proof}
We view the problem as a degree-dependent link failure problem where a
vulnerable link is considered ``operational" and a non-vulnerable link is considered a ``failure."
In this model, each link with degree $k$ fails with probability $1-\rho_k$.  Applying
Theorem \ref{thm:sufficient} directly, we obtain the result.
\end{proof}

\begin{theorem}
For any $G(\mathcal{H}_{\lambda},1)$ with $\lambda>\lambda_c$, if
\begin{multline}\label{sigma-upper-bound}
\sum_{k=1}^{\infty}\frac{\left(\frac{\lambda'}{2}\right)^k}{k!}e^{-\frac{\lambda'}{2}}
\sum_{m=0}^{\infty}\frac{[\lambda' (2\sqrt{2}+\pi)]^m}{m!}e^{-\lambda'
(2\sqrt{2}+\pi)}\times\\\left(1-\left[1-F_{\psi}\left(\frac{m+k-2}{m+k-1}\right)\right]^k\right)<\frac{1}{27},
\end{multline}
where $\lambda' = \pi\lambda^2/2$ and $F_{\psi}(-\infty)=0$ by convention, 
then with probability 1, there is no infinite component
of unreliable links. As a consequence, with probability 1, there is no cascading link failure in
$G(\mathcal{H}_{\lambda},1)$ no matter where the initial link failure is.
\end{theorem}

\begin{proof}
We apply Theorem~\ref{thm:necessary}. Regard an
unreliable link as ``operational" and a reliable link as a ``failure". Then, $\sigma_k$
becomes the failure probability $q(k)$ in
the context of Theorem~\ref{thm:necessary}. Since $\sigma_k$ is
non-increasing in $k$, we replace $q(m+k-1)$ in~\eqref{qk-lower-bound-2} with
$\sigma_{m+k-1}=1-F_{\psi}\left(\frac{m+k-2}{m+k-1}\right)$ and obtain~\eqref{sigma-upper-bound}.
Thus, by Theorem~\ref{thm:necessary}, when~\eqref{sigma-upper-bound} holds, with
probability 1, there is no infinite component of unreliable links in the network.
Next, using the covering graph and techniques from the proof of Theorem 2(ii) in~\cite{Zhening},
we show that if there is no infinite
component of unreliable links, then with probability 1, there is no cascading link failure no matter
where the initial link failure is.
\end{proof}

\section{{\protect\large {Conclusion}}}

\label{conclusion}

In this paper, we studied the problem of cascading link failures in the power grid from a percolation-based viewpoint.  To reflect the fact that links fail according to the amount of power flow going through them, we introduced a model where links fail according to a probability which depends on the number of neighboring links.  We introduced a mapping which maps links in a random geometric graph to nodes in a corresponding dual covering graph.  This mapping enabled us to obtain the first-known analytical conditions on the existence and non-existence of a large component of operational
links after degree-dependent link failures.   Next, we presented a simple but descriptive model
for cascading link failure, and used the degree-dependent link failure results to obtain the first-known analytical conditions on the 
existence and non-existence of cascading link failures.  In particular, we revealed the important roles played by vulnerable and reliable 
links.

\bibliographystyle{ieeetr}
\bibliography{./Power_Link_Failure,./PercolationTopic2}

\end{document}